\documentclass[twocolumn,letterpaper]{IEEEAerospaceCLS}  

\usepackage[T1]{fontenc}
\usepackage[latin9]{inputenc}
\usepackage{mathrsfs}
\usepackage{amsmath}
\usepackage{amsthm}
\usepackage{amssymb}

\usepackage[]{graphicx}    
\newcommand{\ignore}[1]{}  

\makeatletter

\newcommand{\noun}[1]{\textsc{#1}}

\theoremstyle{plain}
\ifx\thechapter\undefined
\newtheorem{thm}{\protect\theoremname}
\else
\newtheorem{thm}{\protect\theoremname}[chapter]
\fi
\theoremstyle{plain}
\newtheorem{lem}[thm]{\protect\lemmaname}
\theoremstyle{definition}

\theoremstyle{plain}
\newtheorem{prop}[thm]{\protect\propositionname}

\usepackage{balance}

\makeatother

\usepackage{babel}
\providecommand{\definitionname}{Definition}
\providecommand{\lemmaname}{Lemma}
\providecommand{\propositionname}{Proposition}
\providecommand{\theoremname}{Theorem}

\begin{document}
\title{Trajectories for the Optimal Collection of Information}

\author{
Matthew R. Kirchner\\ 
Department of ECE\\
University of California, Santa Barbara\\
Santa Barbara, CA 93106-9560\\
kirchner@ucsb.edu
\and 
David Grimsman\\
Computer Science Department\\
Brigham Young University\\
Provo, UT\\
grimsman@cs.byu.edu\\
\and
Jo{\~a}o P. Hespanha\\
Department of ECE\\
University of California, Santa Barbara\\
Santa Barbara, CA 93106-9560\\
hespanha@ece.ucsb.edu
\and 
Jason R. Marden\\
Department of ECE\\
University of California, Santa Barbara\\
Santa Barbara, CA 93106-9560\\
jmarden@ece.ucsb.edu
}

\maketitle

\thispagestyle{plain}
\pagestyle{plain}

\maketitle

\thispagestyle{plain}
\pagestyle{plain}

\begin{abstract}
We study a scenario where an aircraft has multiple heterogeneous sensors
collecting measurements to track a target vehicle of unknown location.
The measurements are sampled along the flight path and our goals to
optimize sensor placement to minimize estimation error. We select
as a metric the Fisher Information Matrix (FIM), as ``minimizing''
the inverse of the FIM is required to achieve small estimation error.
We propose to generate the optimal path from the Hamilton\textendash Jacobi
(HJ) partial differential equation (PDE) as it is the necessary and
sufficient condition for optimality. A traditional method of lines
(MOL) approach, based on a spatial grid, lends itself well to the
highly non-linear and non-convex structure of the problem induced
by the FIM matrix. However, the sensor placement problem results in
a state space dimension that renders a naive MOL approach intractable.
We present a new hybrid approach, whereby we decompose the state space
into two parts: a smaller subspace that still uses a grid and takes
advantage of the robustness to non-linearities and non-convexities,
and the remaining state space that can be found efficiently from a
system of ODEs, avoiding formation of a spatial grid. 
\end{abstract}

\tableofcontents

\section{Introduction}

We present a method to optimize vehicle trajectories to gain maximal
information for target tracking problems. The scenario currently being
studied is an aircraft receiving passive information from sensors
rigidly mounted to the airframe. These sensors include, but are not
limited to, infrared or visible spectrum, as well as RF receivers
that measure the frequency shifts from an external transmitter. The
measurements are sampled in order to determine the location of a target
vehicle. The placement of the sensors is determined by the path of
the aircraft, influencing how much information is gained as well as
the overall effectiveness of estimating where the target is located.
By optimizing the trajectory, we can achieve maximum information gain,
and hence the greatest accuracy in localizing the target.

This problem is a generalization of what appeared in \cite{kirchner2020heterogeneous},
where the path of the vehicle was fixed and a subset of measurements
were selected only from along this path. In this context we optimize
a metric of the cumulative Fisher Information Matrix (FIM) of the
aircraft path, which is motivated by its connection to the (Bayesian)
Cram\'{e}r-Rao lower bound \cite{gill1995applications}. The \noun{logdet}
metric is chosen as this gives a D-optimal estimate, essentially corresponding
to minimizing the volume of the error ellipsoid, and additionally
provides favorable numeric properties. It is worth noting that while
the focus of this paper is the \noun{logdet} metric, other metrics
may be considered, provided the metric meets certain conditions that
are outlined in what follows in the paper. Of particular interest
would be the trace of the inverse metric, as that gives the A-optimal
estimate, effectively minimizing the mean-square estimate error. Analysis
of the trace of the inverse metric is outside the scope of this paper
and will be investigated in future work.

\begin{figure}
\begin{centering}
\includegraphics[width=8cm]{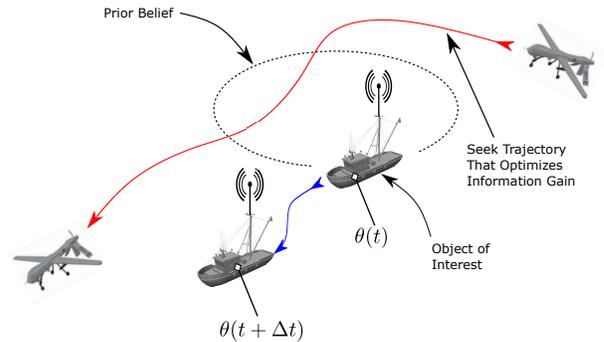}
\par\end{centering}
\caption{An illustration of the target tracking problem. An aircraft collects
measurement for sensors as it flies along a path, attempting to estimate
the location of the ship, denoted here as $\theta$. Modifying the
path of the vehicle can greatly improve the estimation performance.}
\end{figure}
We formulate the problem in such a way that the optimal value function
satisfies a Hamilton-Jacobi (HJ) partial differential equation (PDE),
from which the optimal trajectories immediately follow. Naively, a
solution of the corresponding HJ PDE using a grid-based method would
have many advantages since they handle the non-linear and non-convex
problems that arises in FIM-based optimization. However, the sensor
estimation problem induces a state space dimension that renders typical
grid-based methods \cite{osher2003level} for PDE solutions intractable
due to the exponential dimensional scaling of such methods. Recognition
of this problem is not new, and the phrase ``curse of dimensionality''
was coined decades ago by Richard Bellman \cite{wright1962adaptive}.
This creates a large gap between the rigorous theory of HJ equations
and practical implementation on many problems of interest, especially
vehicle planning and coordination problems.

New research has emerged in an attempt to bridge this technological
gap, including trajectory optimization approaches \cite{darbon2016algorithms,kirchner2017time,kirchner2018primaldual},
machine learning techniques \cite{bertsekas1996neuro,onken2021neural,bansal2021deepreach},
and sub-problem decomposition \cite{chen2018decomposition,kirchner2020hamilton}.
The structure of the sensor placement problem lends itself well to
the later strategy. Unique in this context, though, is that we do
not need to abandon spatial grids entirely, instead forming a hybrid
approach. This leverages the strength of grid-based methods in dealing
with the non-convexities that commonly arise when using the FIM matrix,
but restricts their applications to a small subspace of the problem. 

In what follows we formally introduce the sensor estimation problem
and form its corresponding HJ PDE. We then proceed to show a new hybrid
method of lines (MOL) approach that involves decomposing the state
space. and conclude with simulated results of the optimal trajectories
that result from heterogeneous sensors tracking the location of a
mobile target. Section 2 shows how the information collecting problem
gives rise to nonlinear dynamics with a cascade structure, that the
input only directly affects one first subcomponent of the state, whereas
the optimization criteria only depends on a second subcomponent. Section
3 addresses the optimal control of this type of systems using the
HJ PDE and the classical MOL. Section 4, develops the theory needed
for the new hybrid method of lines, which is applicable to systems
in a cascade form. This type of systems arises naturally in formation
collecting, but the hybrid methods of lines can be applied to the
optimal of more general cascade systems. Section 5 specializes the
hybrid MOL to the information collection. Section 6 includes simulation
results for a particular vehicle model and sensor type.

\section{\label{sec:The-Vehicle-Sensing}The Vehicle Sensing Problem}

We choose as our vehicle a Dubin's car \cite{dubins1957curves} and
denote by $\left(X,Y,\psi\right):=x\in{\mathcal{X}}:={\mathbb{R}}^{2}\times\text{SO}\left(2\right)$
the vehicle state where $X$ and $Y$ are the rectangular positional
coordinates of the vehicle center and $\psi$ is the heading angle.
The dynamics are defined by
\begin{align}
\frac{d}{ds}x\left(s\right) & =f\left(x\left(s\right)\right)+Bu\left(s\right),\,\,\text{a.e.}\,s\in\left[0,t\right]\label{eq:state dynamics}
\end{align}
where
\begin{equation}
f\left(x\right)=\left[\begin{array}{c}
v\cos\psi\\
v\sin\psi\\
0
\end{array}\right],\,B=\left[\begin{array}{c}
0\\
0\\
1
\end{array}\right],\label{eq:f and B def}
\end{equation}
where $u\left(s\right)\in U:=\left[-\omega_{\max},\omega_{\max}\right]$
is the allowable control set of turn rates and $v$ is the fixed forward
speed of the vehicle. The admissible control set is defined as
\begin{equation}
U\left[0,t\right]:=\left\{ u\left(\cdot\right):\left[0,t\right]\rightarrow U\,|\,u\left(\cdot\right)\,\text{is measurable}\right\} .\label{eq:allowable control sequence}
\end{equation}
Our method applied to vehicles that can be expressed in the general
form $\left(\ref{eq:state dynamics}\right)$, which includes the Dubins
vehicle in $\left(\ref{eq:f and B def}\right)$. The Dubins vehicle
with bounded turning rate is particularly interesting because it is
a low-dimensional model that generates trajectories that are easy
to track by an aircraft flying at constant speed and altitude.

The vehicle defined above has a group of rigidly attached sensors
collecting measurements. The measurements, denoted as $y$, are sampled
in order to determine an unknown random variable, $\theta$. The measurements
are assumed to be random variables dependent on $\theta$ with density
function
\[
y\sim\rho\left(y|\theta\right).
\]
Assuming that all measurements $y$ are conditionally independent
given $\theta$, the cumulative Bayesian Fisher Information Matrix
(FIM) associated with the estimation of $\theta$ is of the form
\[
\text{FIM}\left(t,x,u\left(\cdot\right)\right):=Q_{0}+\int_{0}^{t}Q\left(\gamma\left(s;x,u\left(\cdot\right)\right)\right)ds,
\]
where
\begin{equation}
Q\left(x\right):={\mathbb{E}}_{\theta}\left[Q\left(x;\theta\right)\right],\label{eq:Q(x)}
\end{equation}
with
\begin{equation}
Q\left(x;\theta\right):={\mathbb{E}}_{y}\left[\left(\frac{\partial\log\rho\left(y|\theta,x\right)}{\partial\theta}\right)^{\top}\left(\frac{\partial\log\rho\left(y|\theta,x\right)}{\partial\theta}\right)\right],\label{eq:Inner Q(x;theta)}
\end{equation}
and
\[
Q_{0}:={\mathbb{E}}_{\theta}\left[\left(\frac{\partial\log\rho\left(\theta\right)}{\partial\theta}\right)^{\top}\left(\frac{\partial\log\rho\left(\theta\right)}{\partial\theta}\right)\right],
\]
where $\rho\left(\theta\right)$ is the a-priori probability density
function for $\theta$. The formula above assumes a scenario where
the measurement, $y\left(t\right)$, is collected by one sensor or
by multiple independent sensors that generate at the same (constant)
sampling rate. When multiple independent sensors collect measurements
at constant but different sampling rates, the FIM matrix can be factored
for each sensor $i$:
\[
Q\left(t,x,u\left(\cdot\right)\right)=\sum_{i}F^{i}Q^{i}\left(\gamma\left(s;x,u\left(\cdot\right)\right)\right),
\]
where $F^{i}$ is the sampling rate of the $i$-th sensor. The above
matrices are given from \cite{shirazi2019bayesian}, where the expectation
over $y$ in $\left(\ref{eq:Inner Q(x;theta)}\right)$ is given in
closed form for some distributions, see for example \cite[Sec. 5]{kirchner2020heterogeneous}.
While the outer expectation over $\theta$ in $\left(\ref{eq:Q(x)}\right)$
is rarely known in closed form, many approximation schemes can be
employed, for example Monte Carlo sampling or Taylor series expansion.

The placement of the sensors is determined by the path of the aircraft,
influencing how much information is gained as well the overall effectiveness
of estimating $\theta$. Therefore we optimize the trajectory to achieve
maximum information gain, and hence the greatest performance in estimating
$\theta$ from the measurements $y$. For a given initial state $x\in{\mathcal{X}}$
and terminal time $t\in\left[0,\infty\right)$, we define the following
cost functional:
\begin{equation}
J\left(t,x,u\left(\cdot\right)\right):=G\left(\text{CFIM}\left(t,x,u\left(\cdot\right)\right)\right)+\log\det\left(Q_{0}\right),\label{eq:original metric}
\end{equation}
where 
\[
G\left(x,z\right):=-\log\det\left(\text{vec}^{-1}z\right),
\]
We denote by $V\left(t,x\right)$ the value function defined as 
\begin{align}
V\left(t,x\right) & =\underset{u\left(\cdot\right)\in U\left[0,t\right]}{\text{inf}}\,J\left(t,x,u\left(\cdot\right)\right),\label{eq:original problem}
\end{align}
which can be interpreted as the maximal information gain for a family
of trajectory optimization problems parameterized by initial state
$x\in{\mathcal{X}}$ and terminal time $t\in\left[0,\infty\right)$.

The cost functional in $\left(\ref{eq:original problem}\right)$ is
not in a standard form, so we convert the problem into a common standard,
the so-called Mayer form. To do this we augment the state vector with
 $z\in{\mathcal{Z}}:=\text{dom}\left(G\right)$.
Our new state becomes
\[
\chi:=\left(x,z\right)^{\top},
\]
with augmented dynamics
\begin{equation}
\frac{d}{ds}\chi\left(s\right)=\hat{f}\left(\chi\left(s\right),u\left(s\right)\right)=\left[\begin{array}{c}
f\left(x\left(s\right)\right)\\
\ell\left(x\left(s\right)\right)
\end{array}\right]+\left[\begin{array}{c}
B\\
{\bf 0}
\end{array}\right]u\left(s\right),\label{eq:augmented dynamics}
\end{equation}
with
\[
\ell\left(x\left(s\right)\right):=\text{vec}\left(Q\left(x\left(s\right)\right)\right),
\]
where $\text{vec}$ is the vectorize operator that reshapes a matrix
into a column vector and ${\bf 0}$ is a vector of zeros of the same
number of elements as the augmented variable $z$. If we fix the $z$
initial condition such that
\begin{equation}
z=\text{vec}\left(Q_{0}\right),\label{eq:fix z}
\end{equation}
then the cost functional $\left(\ref{eq:original metric}\right)$
can equivalently written as
\begin{equation}
J\left(t,x,u\left(\cdot\right)\right)=J\left(t,\chi,u\left(\cdot\right)\right)=G\left(\text{vec}^{-1}\left(z\right)\right),\label{eq:New J with G(vec)}
\end{equation}
where we denote by $Z=\text{vec}^{-1}\left(z\right)$ the inverse
operator such that
\[
\text{vec}\left(\text{vec}^{-1}\left(z\right)\right)=z.
\]
Hereafter we will denote by $\tilde{G}$ as the function $G$ with
the input reshaped as a function of $z$ with
\begin{equation}
\tilde{G}\left(z\right):=G\left(\text{vec}^{-1}\left(z\right)\right).\label{eq:vec of G}
\end{equation}
Likewise the value function is equivalently written as
\begin{equation}
V\left(t,\chi\right)=\underset{u\left(\cdot\right)\in U\left[0,t\right]}{\text{inf}}\,J\left(t,\chi,u\left(\cdot\right)\right).\label{eq:new value function}
\end{equation}

\section{\label{sec:Decomposition-of-Coupled}Decomposition of Coupled Systems}

The approach we will develop to solve $\left(\ref{eq:new value function}\right)$
is applicable to a more general class of cascade systems that we introduce
in this section, and for which we discuss the use of HJ methods for
optimal control. Denote by $\chi:=\left(x,z\right)^{\top}$ where
$x\in{\mathcal{X}}={\mathbb{R}}^{n}$ and $z\in{\mathcal{Z}}={\mathbb{R}}^{m}$.
The state has coupled dynamics as follows:
\begin{equation}
\begin{cases}
\dot{x}\left(s\right)=f\left(x\left(s\right)\right)+g\left(x\left(s\right)\right)u\left(s\right) & \text{a.e}\,s\in\left[0,t\right]\\
\dot{z}\left(s\right)=\ell\left(x\left(s\right)\right),
\end{cases}\label{eq: Coupled Gen. System}
\end{equation}
with $u\in U$, where $U$ is a closed convex set. We denote by $\left[0,t\right]\ni s\mapsto\gamma\left(s;x_{0},u\left(\cdot\right)\right)\in{\mathbb{R}}^{n}$
the $x$ state trajectory that evolves in time according to $\left(\ref{eq:state dynamics}\right)$
starting from initial state $x_{0}$ at $t=0$. The trajectory $\gamma$
is a solution of $\left(\ref{eq:state dynamics}\right)$ in that it
satisfies $\left(\ref{eq:state dynamics}\right)$ almost everywhere:
\begin{equation}
\begin{cases}
\dot{\gamma}\left(s;x_{0},u\left(\cdot\right)\right)=f\left(\gamma\left(s;x_{0},u\left(\cdot\right)\right)\right)+g\left(\gamma\left(s;x_{0},u\left(\cdot\right)\right)\right)u,\\
\gamma\left(0;x_{0},u\left(\cdot\right)\right)=x_{0}.
\end{cases}\label{eq:dynamic constraints}
\end{equation}
Likewise, we denote by $\left[0,t\right]\ni s\mapsto\xi\left(s;\chi_{0},u\left(\cdot\right)\right)$
the trajectory of the $z$ variable and it satisfies the following
almost everywhere:
\begin{equation}
\begin{cases}
\frac{d}{ds}\xi\left(s;\chi_{0},u\left(\cdot\right)\right)=\ell\left(\gamma\left(s;x_{0},u\left(\cdot\right)\right)\right),\\
\xi\left(0;\chi_{0},u\left(\cdot\right)\right)=z_{0}.
\end{cases}\label{eq:aux variable trajectory}
\end{equation}
Note that the trajectory can be found directly from the expression:
\begin{equation}
\xi\left(s;\chi_{0},u\left(\cdot\right)\right):=z_{0}+\int_{0}^{s}\ell\left(\gamma\left(\tau;x_{0},u\left(\cdot\right)\right)\right)d\tau.\label{eq:running cost as trajectory-1}
\end{equation}
Denote $G:{\mathbb{R}}^{m}\rightarrow{\mathbb{R}}$ as the terminal cost
function such that the mapping
\[
{\mathcal{Z}}\ni z\mapsto G\left(z\right)\in{\mathbb{R}},
\]
We define the cost functional
\[
J\left(t,\chi,u\left(\cdot\right)\right):=G\left(\xi\left(t;\chi,u\left(\cdot\right)\right)\right),
\]
and the associated value function as
\[
V\left(t,\chi\right):=\underset{u\left(\cdot\right)\in U\left[0,t\right]}{\text{inf}}J\left(t,\chi,u\left(\cdot\right)\right),
\]
where $U\left[0,t\right]$ is defined as in $\left(\ref{eq:allowable control sequence}\right)$.

We denote by
\[
\hat{f}\left(\chi,u\right):=\left[\begin{array}{c}
f\left(x\right)+g\left(x\right)u\\
\ell\left(x\right)
\end{array}\right],
\]
the joint vector field in $\left(\ref{eq: Coupled Gen. System}\right)$.
We assume that $\hat{f}$ , $U$, and $G$ satisfy the following regularity
assumptions:
\begin{description}
\item [{(F1)}] $\left(U,d\right)$ is a separable metric space.
\item [{(F2)}] The maps $\hat{f}:{\mathcal{X}}\times U\rightarrow{\mathbb{R}}^{n+m}$
and $G:{\mathcal{Z}}\rightarrow{\mathbb{R}}$ are measurable, and there
exists a constant $L>0$ and a modulus of continuity $\omega:\left[0,\infty\right)\rightarrow\left[0,\infty\right)$
such that for $\varphi\left(\chi,u\right)=\hat{f}\left(\chi,u\right),G\left(z\right)$,
we have for all $\chi,\chi'\in{\mathcal{X}}\times{\mathcal{Z}}$, and
$u,u'\in U$ 
\[
\left|\varphi\left(\chi,u\right)-\varphi\left(\chi',u'\right)\right|\leq L\left\Vert \chi-\chi'\right\Vert +\omega\left(d\left(u,u'\right)\right),
\]
and
\[
\left|\varphi\left({\bf 0},u\right)\right|\leq L.
\]
\item [{(F3)}] The maps $\hat{f}$, and $G$ are $C^{1}$ in $\chi$, and
there exists a modulus of continuity $\omega:\left[0,\infty\right)\rightarrow\left[0,\infty\right)$
such that for $\varphi\left(\chi,u\right)=\hat{f}\left(\chi,u\right),G\left(z\right)$,
we have for all $\chi,\chi'\in\mathcal{X}\times\mathcal{Z}$, and
$u,u'\in U$
\[
\left|\varphi_{\chi}\left(\chi,u\right)-\varphi_{\chi}\left(\chi',u'\right)\right|\leq\omega\left(\left\Vert \chi-\chi'\right\Vert +d\left(u,u'\right)\right).
\]
\end{description}

\subsection{Hamilton\textendash Jacobi Formulation}

Under a set of mild Lipschitz continuity assumptions, there exists
a unique value function $\left(\ref{eq:new value function}\right)$
that satisfies the following Hamilton\textendash Jacobi (HJ) equation
\cite{evans10} with $V\left(t,\chi\right)$ being the viscosity solution
of the partial differential equation (PDE) for $s\in\left[0,t\right]$
\begin{align}
V_{s}\left(s,\chi\right)+{\mathcal{H}}\left(\chi,V_{\chi}\left(s,\chi\right)\right) & =0,\label{eq:HJB Equation-1}\\
V\left(0,\chi\right) & =G\left(z\right),\nonumber 
\end{align}
where $\sigma:=\left(p,\lambda\right)^{\top}$ and
\begin{equation}
{\mathcal{H}}\left(\chi,\sigma\right):=\underset{u\in U}{\min}H\left(\chi,u,\sigma\right),\label{eq:Optimal Hamiltonian-1}
\end{equation}
with the Hamiltonian, $H$, defined by
\begin{align*}
H\left(\chi,u,\sigma\right) & =\left\langle \left[\begin{array}{c}
f\left(x\right)+g\left(x\right)u\\
\ell\left(x\right)
\end{array}\right],\left[\begin{array}{c}
p\\
\lambda
\end{array}\right]\right\rangle \\
 & =\left\langle f\left(x\right),p\right\rangle +\left\langle g\left(x\right)u,p\right\rangle +\left\langle \ell\left(x\right),\lambda\right\rangle .
\end{align*}
In the case where the set $U$ is bounded by a norm, i.e.
\begin{equation}
U=\left\{ u\in{\mathbb{R}}^{n_{u}}|\left\Vert u\right\Vert \leq c\right\} ,\label{eq:general control set}
\end{equation}
for some $c$, then $\left(\ref{eq:Optimal Hamiltonian-1}\right)$
is given in closed form by
\begin{equation}
{\mathcal{H}}\left(\chi,\rho\right)=\left\langle f\left(x\right),p\right\rangle +\left\Vert g\left(x\right)^{\top}p\right\Vert _{*}+\left\langle \ell\left(x\right),\lambda\right\rangle ,\label{eq:optimal Hamiltonian-1}
\end{equation}
where $\left\Vert \left(\cdot\right)\right\Vert _{*}$ is the dual
norm to $\left\Vert \left(\cdot\right)\right\Vert $ in $\left(\ref{eq:general control set}\right)$.
We denote by $\pi$ the control that optimizes the Hamiltonian and
is given by
\[
\pi\left(s,\chi\right):=\underset{u\in U}{\arg\,\min}H\left(\chi,u,V_{\chi}\left(s,\chi\right)\right).
\]
We note here that under mild assumptions, the viscosity solution of
$\left(\ref{eq:HJB Equation-1}\right)$ is Lipschitz continuous in
both $s$ and $\chi$ \cite[Theorem 2.5, p. 165]{yong1999stochastic}.
This implies by Rademacher's theorem \cite[Theorem 3.1.6, p. 216]{federer2014geometric}
the value function is differentiable almost everywhere. For what follows,
we assume that the value function has continuous first and second
derivatives. The points where this fails to be true only exists on
a set of measure zero, and any practical implementation of the method
presented will only evaluate points where the first and second derivatives
exist. A characterization of the differentiability of the value function
is outside the scope of this paper and a full rigorous treatment will
appear in forthcoming work. 

\subsection{Necessary Conditions of the Optimal Trajectories}

Fix $x\in{\mathcal{X}}$ and $z\in{\mathcal{Z}}$ as initial conditions
and fix the terminal time $t$. Denote by $\bar{\gamma}\left(s\right)$
and $\bar{\xi}\left(s\right)$ as the optimal state trajectories such
that
\[
\bar{\gamma}\left(s\right):=\bar{\gamma}\left(s;\chi\right)=\gamma\left(s;x,\bar{u}\left(\cdot;\chi\right)\right),
\]
and
\[
\bar{\xi}\left(s\right):=\bar{\xi}\left(s;\chi\right)=\xi\left(s;x,z,\bar{u}\left(\cdot;\chi\right)\right),
\]
such that $\bar{u}$ optimizes $\left(\ref{eq:new value function}\right)$.
By Pontryagin's theorem \cite{pontryagin2018mathematical} there exists
adjoint trajectories $p\left(s\right):=p\left(s;\chi\right)$ and
$\lambda\left(s\right):=\lambda\left(s;\chi\right)$ such that the
function 
\begin{equation}
\left[0,t\right]\ni s\mapsto\left(\bar{\gamma}\left(s\right),\bar{\xi}\left(s\right),p\left(s\right),\lambda\left(s\right)\right)\label{eq: Optimal trajectories}
\end{equation}
is a solution of the characteristic system
\begin{equation}
\begin{cases}
\dot{\bar{\gamma}}\left(s\right)={\mathcal{H}}_{p}\left(\bar{\gamma}\left(s\right),\bar{\xi}\left(s\right),p\left(s\right),\lambda\left(s\right)\right),\\
\dot{\bar{\xi}}\left(s\right)={\mathcal{H}}_{\lambda}\left(\bar{\gamma}\left(s\right),\bar{\xi}\left(s\right),p\left(s\right),\lambda\left(s\right)\right),\\
\dot{p}\left(s\right)=-{\mathcal{H}}_{x}\left(\bar{\gamma}\left(s\right),\bar{\xi}\left(s\right),p\left(s\right),\lambda\left(s\right)\right),\\
\dot{\lambda}\left(s\right)=-{\mathcal{H}}_{z}\left(\bar{\gamma}\left(s\right),\bar{\xi}\left(s\right),p\left(s\right),\lambda\left(s\right)\right),
\end{cases}\label{eq:characteristic system}
\end{equation}
almost everywhere $s\in\left[0,t\right]$ with boundary conditions
\[
p\left(t\right)={\bf 0},\,\,\lambda\left(t\right)=G_{z}\left(\bar{\xi}\left(t\right)\right).
\]

\subsection{Numerical Approximations Viscosity Solutions to First-Order Hyperbolic
PDEs}

Traditional methods for computing the viscosity solution to $\left(\ref{eq:HJB Equation-1}\right)$
rely on constructing a discrete grid of points. This is typically
chosen as a Cartesian grid, but many other grid types exist. The value
function is found using a \emph{method of lines }(MOL)\emph{ }approach
by the solving the following family of ODEs, pointwise at each grid
point $\chi^{k}=\left(x^{k},z^{k}\right)\in{\mathcal{S}}:={\mathcal{X}}\times{\mathcal{Z}}$:
\begin{equation}
\begin{cases}
\dot{\phi}\left(s,\chi^{k}\right)=-{\mathcal{H}}\left(\chi^{k},D_{\chi}\phi\left(s,\chi^{k}\right)\right), & s\in\left[0,t\right]\\
\phi\left(0,\chi^{k}\right)=G\left(z^{k}\right),
\end{cases}\label{eq:method of lines family of odes}
\end{equation}
 where $\phi\left(s,\chi^{k}\right)$ should be viewed as an approximation
to the value function $V\left(s,\chi^{k}\right)$ in $\left(\ref{eq:HJB Equation-1}\right)$
and
\[
D_{\chi}\phi\left(s,\chi^{k}\right)\approx\phi_{\chi}\left(s,\chi^{k}\right)
\]
is obtained by a finite difference scheme used to approximate the
gradient of $\phi$ at grid point $k$. Care must be taken when evaluating
finite differences of possibly non-smooth functions and the family
of \emph{Essentially Non-Oscillatory} (ENO) methods were developed
to address this issue \cite{jiang2000weighted}. The advantage of
the method of lines is that we can compute $\left(\ref{eq:method of lines family of odes}\right)$
independently at each grid point with $\phi\left(t,\chi^{k}\right)\approx V\left(t,\chi^{k}\right)$.
Under certain conditions, for example the Lax-Richtmyer equivalence
theorem \cite{lax1956survey}, 
\[
\Delta s\rightarrow0,\,\Delta\chi\rightarrow0\implies\phi\left(t,\chi^{k}\right)\rightarrow V\left(t,\chi^{k}\right)
\]
when the scheme is both consistent, i.e. the error between $\phi\left(t,\chi^{k}\right)$
and $V\left(t,\chi^{k}\right)$ tends to zero, and stable. In this
case, stability is enforced when the time step, $\Delta s$, satisfies
the Courant-Friedrichs-Lewy (CFL) condition \cite{courant1967partial}.
When the HJ equation is a non-linear PDE, then additionally a Lax-Friedrichs
approximation \cite{crandall1984two,osher1991high} is needed to ensure
stability. In the Lax-Friedrichs method the Hamiltonian in $\left(\ref{eq:method of lines family of odes}\right)$
is replaced by
\begin{align*}
\hat{\mathcal{H}}\left(\chi,\sigma^{+},\sigma^{-}\right):= & \mathcal{H}\left(\chi,\frac{\sigma^{+}+\sigma^{-}}{2}\right)\\
 & -\nu\left(\chi\right)^{\top}\left(\frac{\sigma^{+}+\sigma^{-}}{2}\right),
\end{align*}
where inputs $D_{\chi}^{+}\phi\left(s,\chi^{k}\right)\rightarrow\sigma^{+}$
and $D_{\chi}^{-}\phi\left(s,\chi^{k}\right)\rightarrow\sigma^{-}$
are the right and left side bias finite differencing approximations
to the gradient, respectively. The term $\nu\left(\chi\right)$ is
the artificial dissipation and depends on $H_{\sigma}\left(\chi,\sigma\right)$,
the gradient of the Hamiltonian with respect to the adjoint variable.
The MOL approach in $\left(\ref{eq:method of lines family of odes}\right)$
becomes
\begin{equation}
\begin{cases}
\dot{\phi}\left(s,\chi^{k}\right)=-\hat{\mathcal{H}}\left(\chi^{k},D_{\chi}^{+}\phi\left(s,\chi^{k}\right),D_{\chi}^{-}\phi\left(s,\chi^{k}\right)\right),\\
\phi\left(0,\chi^{k}\right)=G\left(z^{k}\right),
\end{cases}\label{eq:method of lines family of odes with Lax-Friedrichs}
\end{equation}

In general, no closed form solution exists to $\left(\ref{eq:method of lines family of odes with Lax-Friedrichs}\right)$
and therefore an explicit Runge-Kutta scheme is employed. If the first
order Euler method is used to solve $\left(\ref{eq:method of lines family of odes with Lax-Friedrichs}\right)$,
then we have the following time-marching scheme with iteration for
$s\in\left[0,t\right]$:
\begin{equation}
\begin{cases}
\phi\left(s+\Delta s,\chi^{k}\right)=\phi\left(s,\chi^{k}\right)\\
\,\,\,\,\,\,\,\,\,\,\,\,\,\,\,\,\,\,\,\,\,\,\,\,\,\,\,\,\,\,\,\,\,\,\,\,\,\,\,-\Delta s\hat{\mathcal{H}}\left(\chi^{k},D_{\chi}^{+}\phi\left(s,\chi^{k}\right),D_{\chi}^{-}\phi\left(s,\chi^{k}\right)\right),\\
\phi\left(0,\chi^{k}\right)=G\left(z^{k}\right).
\end{cases}\label{eq: method of lines, approx}
\end{equation}
 The reader is encouraged to read \cite{osher2003level} for a comprehensive
review on numeric numeric methods to solving first-order hyperbolic
HJ PDEs.

\section{\label{sec:HJB-Decomposition}HJB Decomposition}

We are especially interested in problems for which the $x$-component
of the state in $\left(\ref{eq: Coupled Gen. System}\right)$ has
a relatively small dimension, but $z$-component does not. This is
common in the vehicle sensing problem discussed in Section \ref{sec:The-Vehicle-Sensing},
because the dimension of $z$ scales with the square of the number
of parameters to be estimated and therefore, even for simple vehicle
dynamics and a relatively small number of parameters, the dimension
of the state $\chi$ is too large to apply $\left(\ref{eq: method of lines, approx}\right)$.
To overcome this challenge, we present an hybrid method of lines that
uses a grid over $x$, but no grid over $z$.

A key challenge to creating such a method is to find a closed-form
expression for the gradient of the value function with respect to
$z$, so as to avoid finite differencing schemes in $z$. Taking advantage
of the specific structure of the problem, we show that we can use
a grid over the state variable $x$ to compute $D_{x}\phi\left(s,\chi^{k}\right)\approx\phi_{x}\left(s,\chi^{k}\right)$
with finite differences, but avoid a grid over the state variable
$z$ by solving a family of ODEs to compute $D_{z}\phi\left(s,\chi^{k}\right)$.
This is supported by the following theorem.
\begin{thm}
\label{thm: ODE for grad of z}Suppose the value function $V\left(s,\chi\right)$
is twice differentiable at $\left(s,\chi\right)\in\left[0,\infty\right)\times{\mathcal{S}}$.
Then at any point $\chi$, the gradient of the value function with
respect to $z$ can be found using the following ODE:
\begin{equation}
\begin{cases}
\dot{V}_{z}\left(s,\chi\right)=-\frac{\partial}{\partial z}\left\langle G_{z}\left(\bar{\xi}\left(s\right)\right),\ell\left(x\right)\right\rangle \\
\,\,\,\,\,\,\,\,\,\,\,\,\,\,\,\,\,\,\,\,\,\,\,\,\,\,\,\,\,-R_{x}\left(s,\chi,\pi\left(s,\chi\right),f\left(x\right),g\left(x\right)\right),\\
V_{z}\left(0,\chi\right)=G_{z}\left(z\right),
\end{cases}\label{eq:1st term in z ode}
\end{equation}
where
\begin{align}
R_{x}\left(s,\chi,u,\alpha,\beta\right):= & \frac{\partial}{\partial x}\Big\{\left\langle G_{z}\left(\bar{\xi}\left(s\right)\right),\alpha\right\rangle \label{eq:R from decomp}\\
 & +\left\langle G_{z}\left(\bar{\xi}\left(s\right)\right),\beta u\right\rangle \Big\}.
\end{align}
\end{thm}
The proof of Theorem \ref{thm: ODE for grad of z} will need the following
technical lemma.
\begin{lem}
\label{lem: First gradient Vz}Suppose that the gradient $V_{z}\left(t,\chi\right)$
exists at $\left(t,\chi\right)\in\left[0,\infty\right)\times{\mathcal{S}}.$
Then the gradient of the value function with respect to the augmented
variable is given by
\[
V_{z}\left(t,\chi\right)=G_{z}\left(\bar{\xi}\left(t;\chi\right)\right).
\]
\end{lem}
\begin{proof}
Recall from $\left(\ref{eq:running cost as trajectory-1}\right)$
and applying the optimal control sequence,
\[
\bar{\xi}\left(s\right)=z+\int_{0}^{s}\ell\left(\bar{\gamma}\left(\tau\right)\right)d\tau.
\]
Therefore
\begin{equation}
G_{z}\left(z+\int_{0}^{t}\ell\left(\bar{\gamma}\left(\tau\right)\right)d\tau\right)=G_{z}\left(\bar{\xi}\left(t\right)\right):=\lambda\left(t\right).\label{eq:boudary of lambda grad}
\end{equation}
Recognize that $\left(\ref{eq:boudary of lambda grad}\right)$ is
the boundary condition of the characteristic system $\left(\ref{eq:characteristic system}\right)$,
and that
\begin{align*}
V_{z}\left(t,\chi\right) & =\lambda\left(0\right)\\
 & =G_{z}\left(\bar{\xi}\left(t\right)\right)-\int_{t}^{0}{\mathcal{H}}_{z}\left(\bar{\gamma}\left(s\right),\bar{\xi}\left(s\right),p\left(s\right),\lambda\left(s\right)\right)ds.
\end{align*}
Where the first line above uses the connection between the adjoint
variable, $\lambda$, and the value function \cite[Theorem 3.4, p. 235]{yong1999stochastic}.
Observing that the Hamiltonian $\left(\ref{eq:optimal Hamiltonian-1}\right)$
does not depend on the argument $z$, then it follows that 
\[
{\mathcal{H}}_{z}\left(\bar{\gamma}\left(s\right),\bar{\xi}\left(s\right),p\left(s\right),\lambda\left(s\right)\right)=0,\,s\in\left[0,t\right],
\]
which leads to
\[
V_{z}\left(t,\chi\right)=G_{z}\left(\bar{\xi}\left(t\right)\right).
\]
\end{proof}
We now proceed to the proof of Theorem \ref{thm: ODE for grad of z}.
\begin{proof}
Fix $x,z$ and noting the original HJB equation $\left(\ref{eq:HJB Equation-1}\right)$:
\begin{align*}
\dot{V}_{z}\left(s,\chi\right) & =\frac{\partial}{\partial s}\left\{ V_{z}\left(s,\chi\right)\right\} \\
 & =\frac{\partial}{\partial z}\left\{ V_{s}\left(s,\chi\right)\right\} \\
 & =\frac{\partial}{\partial z}\left\{ -{\mathcal{H}}\left(\chi,V_{x}\left(s,\chi\right),V_{z}\left(s,\chi\right)\right)\right\} .
\end{align*}
From the definition of the Hamiltonian
\begin{align*}
\dot{V}_{z}\left(s,\chi\right)=\frac{\partial}{\partial z}\Bigg\{ & -\left\langle V_{z}\left(s,\chi\right),\ell\left(x\right)\right\rangle -\left\langle V_{x}\left(s,\chi\right),f\left(x\right)\right\rangle \\
 & -\underset{u\in U}{\min}\left\langle V_{x}\left(s,\chi\right),g\left(x\right)u\right\rangle \Bigg\}.
\end{align*}
Fix time $s\in\left[0,t\right]$, and define the function
\begin{align*}
\varphi^{s}\left(\chi,u\right): & =\underset{u\in U}{\min}\,F^{s}\left(\chi,u\right),
\end{align*}
where
\[
F^{s}\left(\chi,u\right):=\left\langle V_{x}\left(s,\chi\right),g\left(x\right)u\right\rangle ,
\]
and recall that 
\[
\pi\left(s,\chi\right):=\underset{u\in U}{\arg\,\min}\left\langle V_{x}\left(s,\chi\right),g\left(x\right)u\right\rangle .
\]
Since by assumption both $V_{x}\left(s,\chi\right)$ and $V_{zx}\left(s,\chi\right)$
exist, and $F^{s}\left(\chi,u\right)$ is differentiable at $\chi$,
this implies the gradient of $\varphi^{s}$ can by found \cite[Theorem 4.13]{bonnans2000perturbation}
with the following relation:
\[
\varphi_{z}^{s}\left(\chi,u\right)=F_{z}^{s}\left(\chi,\pi\left(s,\chi\right)\right).
\]
This gives
\begin{align*}
\dot{V}_{z}\left(s,\chi\right)= & -\frac{\partial}{\partial z}\left\{ \left\langle V_{z}\left(s,\chi\right),\ell\left(x\right)\right\rangle \right\} \\
 & -\frac{\partial}{\partial z}\left\{ \left\langle V_{x}\left(s,\chi\right),\alpha\right\rangle \right\} \biggr\rvert_{\alpha=f\left(x\right)}\\
 & -\frac{\partial}{\partial z}\left\{ \left\langle V_{x}\left(s,\chi\right),\beta u\right\rangle \right\} \biggr\rvert_{u=\pi\left(s,\chi\right),\beta=g\left(x\right)}.
\end{align*}
Noting the symmetry of the gradients with respect to $x,z$ we have
\begin{align*}
\dot{V}_{z}\left(s,\chi\right)= & -\frac{\partial}{\partial z}\left\{ \left\langle V_{z}\left(s,\chi\right),\ell\left(x\right)\right\rangle \right\} \\
 & -\frac{\partial}{\partial x}\left\{ \left\langle V_{z}\left(s,\chi\right),\alpha\right\rangle \right\} \biggr\rvert_{\alpha=f\left(x\right)}\\
 & -\frac{\partial}{\partial x}\left\{ \left\langle V_{z}\left(s,\chi\right),\beta u\right\rangle \right\} \biggr\rvert_{u=\pi\left(s,\chi\right),\beta=g\left(x\right)},
\end{align*}
and then applying Lemma \ref{lem: First gradient Vz}, the result
follows.
\end{proof}

\subsection{Method of Lines with State Space Decomposition}

Recall that we denote by $\phi\left(s,\chi\right)$ the numeric approximation
to the value function, $V\left(s,\chi\right)$. The proposed hybrid
MOL is relies on an approximations $D_{x}\phi\left(s,\chi\right)$
of the gradient of the value function with respect to $x$, $V_{x}\left(s,\chi\right)$,
that is based on the Lax-Friedrichs approximation. However, the approximation
$\Phi\left(s,\chi\right)$ of the gradient of the value function with
respect to $z$, $V_{z}\left(s,\chi\right)$, is obtained by solving
an ODE in time and does not require a spatial grid. In view of this,
this method computes the two approximations $\phi\left(s,x^{k},z\right)$
and $\Phi\left(s,x^{k},z\right)$ on points $\left(x^{k},z\right)\in{\mathcal{S}}$
where the $x^{k}$ are restricted to a finite grid of the $x$-component
of the state, whereas $z$ is not restricted to a grid. To accomplish
this, we need the following assumption that, together with Theorem
1, leads to the following MOL. 

Suppose that the first term in $\left(\ref{eq:1st term in z ode}\right)$
can be written as 
\begin{equation}
\frac{\partial}{\partial z}\left\{ \left\langle G_{z}\left(\bar{\xi}\left(s\right)\right),\ell\left(x\right)\right\rangle \right\} =\Upsilon\left(x,z,G_{z}\left(\bar{\xi}\left(s\right)\right)\right),\label{eq: second gradient as a function of dz}
\end{equation}
and fix $z$ for any $z\in{\mathcal{Z}}$. Denote by $\Phi\left(s,x^{k},z\right)\approx\phi_{z}\left(s,x^{k},z\right)=G_{z}\left(\bar{\xi}\left(s\right)\right)$
as the gradient estimate of the value function with respect to $z$.
Then from Theorem \ref{thm: ODE for grad of z} and Lemma \ref{lem: First gradient Vz},
we construct the following method of lines approach, for $\left(x^{k},z\right)\in{\mathcal{S}}$:
\begin{equation}
\begin{cases}
\dot{\phi}\left(s,x^{k},z\right)=-\tilde{\mathcal{H}}\Big(x^{k},z,D_{x}^{+}\phi\left(s,x^{k},z\right),D_{x}^{-}\phi\left(s,x^{k},z\right),\\
\,\,\,\,\,\,\,\,\,\,\,\,\,\,\,\,\,\,\,\,\,\,\,\,\,\,\,\,\,\,\,\,\,\,\,\,\,\,\,\,\,\,\Phi\left(s,x^{k},z\right)\Big),\\
\dot{\Phi}\left(s,x^{k},z\right)=-\Upsilon\left(x^{k},z,\Phi\left(s,x^{k},z\right)\right)\\
\,\,\,\,\,\,\,\,\,\,\,\,\,\,\,\,\,\,-R_{x}\left(s,x^{k},z,\pi\left(s,x^{k},z\right),f\left(x^{k}\right),g\left(x^{k}\right)\right),\\
\phi\left(0,x^{k},z\right)=G\left(z\right),\\
\Phi\left(0,x^{k},z\right)=G_{z}\left(z\right),
\end{cases}\label{eq:Hybrid MOL}
\end{equation}
where 
\begin{align*}
\tilde{\mathcal{H}}\left(x,z,\rho^{+},\rho^{-},\lambda\right):= & {\mathcal{H}}\left(x,z,\frac{\rho^{+}+\rho^{-}}{2},\lambda\right)\\
 & -\nu\left(x\right)^{\top}\left(\frac{\rho^{+}+\rho^{-}}{2}\right),
\end{align*}
is the Lax-Friedrichs approximation. The Lax-Friedrichs approximation
is only needed in the $x$ dimension since that is the only space
where a grid is constructed for computing finite differences.

\section{Optimal Information Collection}

Recall that the system $\left(\ref{eq:augmented dynamics}\right)$
presented in Section \ref{sec:The-Vehicle-Sensing} is of the form
of Section \ref{sec:Decomposition-of-Coupled}, and we can use Theorem
\ref{thm: ODE for grad of z} to construct a method of lines. Recall
that for Dubins car, $U=\left[-\omega_{\text{max}},\omega_{\text{max}}\right]$,
and the optimal Hamilton $\left(\ref{eq:Optimal Hamiltonian-1}\right)$
becomes
\[
{\mathcal{H}}\left(x,z,p,\lambda\right)=\left\langle f\left(x\right),p\right\rangle +\omega_{\max}\left|B^{\top}p\right|+\left\langle \lambda,\text{vec}\left(Q\left(x\right)\right)\right\rangle ,
\]
and optimal control policy is given by
\begin{align}
\pi\left(s;x,z\right):= & \underset{u\in U}{\arg\,\min}H\left(x,z,u,V_{x}\left(s,x,z\right),V_{z}\left(s,x,z\right)\right)\nonumber \\
 & \in\begin{cases}
-\omega_{\max} & B^{\top}V_{x}\left(s,x,z\right)<0\\
\left[-\omega_{\max},\omega_{\max}\right] & B^{\top}V_{x}\left(s,x,z\right)=0\\
\omega_{\max} & B^{\top}V_{x}\left(s,x,z\right)>0.
\end{cases}\label{eq:optimal policy}
\end{align}
In order to compute the first term in $\left(\ref{eq:1st term in z ode}\right)$
for the vehicle tracking problem presented in Section \ref{sec:The-Vehicle-Sensing},
we present the following lemma.
\begin{lem}
\label{lem:dz trace as a function of gradient}Let $\chi\in{\mathcal{S}}$.
When $G\left(z\right)=-\text{log\,det}\left(\text{vec}^{-1}\left(z\right)\right)$
and $\ell\left(x\right)=\text{vec}\left(Q\left(x\right)\right)$,
then
\begin{align*}
 & \frac{\partial}{\partial z}\left\langle G_{z}\left(\bar{\xi}\left(s\right)\right),\ell\left(x\right)\right\rangle \\
 & =\text{vec}\left(\text{vec}^{-1}\left(G_{z}\left(\bar{\xi}\left(s\right)\right)\right)\cdot Q\left(x\right)\cdot\text{vec}^{-1}\left(G_{z}\left(\bar{\xi}\left(s\right)\right)\right)\right){}^{\top}.
\end{align*}
\end{lem}
\begin{proof}
Define $\bar{\Xi}\left(z\right):=\Xi\left(s;x,z,\bar{u}\left(\cdot\right)\right)=\text{vec}^{-1}\left(\bar{\xi}\left(s;\chi,\bar{u}\left(\cdot\right)\right)\right)$
as the optimal auxiliary state trajectory at the time, $s$, reshaped
into a matrix. The matrix forms simplifies the following proof and
the computations in the examples to follow. We also denote by $Z:=\text{vec}^{-1}\left(z\right)$.
The gradient with respect to a matrix of a function $F\left(Z\right)$
is the matrix defined by
\[
\frac{\partial}{\partial Z}F\left(Z\right):=\text{vec}^{-1}\left\{ \left[\frac{\partial F\left(Z\right)}{\partial Z_{ij}}\right]_{i,j}\right\} .
\]
Recall $\left(\ref{eq:vec of G}\right)$ and from Lemma \ref{lem: First gradient Vz}
that $V_{z}\left(s,x,z\right)=G_{z}\left(\bar{\Xi}\left(z\right)\right)=\text{vec}^{-1}\left(\bar{\Xi}\left(z\right)^{-1}\right)$.
Then we have
\begin{align*}
 & \frac{\partial}{\partial z}\left\langle G_{z}\left(\bar{\xi}\left(s\right)\right),\ell\left(x\right)\right\rangle \\
 & =\text{vec}\left(\frac{\partial}{\partial Z}\text{tr}\left(\bar{\Xi}\left(z\right)^{-1}Q\left(x\right)\right)\right).
\end{align*}
We direct our attention to the term inside the $\text{vec}$ operator
in the last line above, and find
\begin{align*}
 & \frac{\partial}{\partial Z_{ij}}\text{tr}\left(\bar{\Xi}\left(z\right)^{-1}Q\left(x\right)\right)\\
 & =\text{tr}\left(\frac{\partial}{\partial Z_{ij}}\left\{ \bar{\Xi}\left(z\right)^{-1}\right\} Q\left(x\right)\right)\\
 & =\text{tr}\left(-\bar{\Xi}\left(z\right)^{-1}\frac{\partial\bar{\Xi}\left(z\right)}{\partial Z_{ij}}\bar{\Xi}\left(z\right)^{-1}Q\left(x\right)\right)
\end{align*}
where the last line is from \cite{petersen2008matrix}. Noting $\frac{\partial\bar{\Xi}\left(z\right)}{\partial Z_{ij}}=\frac{\partial\bar{\Xi}\left(z\right)}{\partial Z}\frac{\partial z}{\partial Z_{ij}}$,
recalling from Proposition \ref{prop:gradient in z with intial cond. is I}
that $\frac{\partial\bar{\Xi}\left(z\right)}{\partial Z}=I$ and noting
$\frac{\partial z}{\partial Z_{ij}}=S^{ij}:=e_{i}e_{j}^{\top}$, where
$e_{k}$ is a vector with a $1$ in $k$-th element and zeros elsewhere.
We now have
\begin{align*}
 & \frac{\partial}{\partial Z_{ij}}\text{tr}\left(\bar{\Xi}\left(z\right)^{-1}Q\left(x\right)\right)\\
 & =-\text{tr}\left(\bar{\Xi}\left(z\right)^{-1}e_{i}e_{j}^{\top}\bar{\Xi}\left(z\right)^{-1}Q\left(x\right)\right)\\
 & =-\text{tr}\left(e_{j}^{\top}\bar{\Xi}\left(z\right)^{-1}Q\left(x\right)\bar{\Xi}\left(z\right)^{-1}e_{i}\right)\\
 & =-e_{j}^{\top}\bar{\Xi}\left(z\right)^{-1}Q\left(x\right)\bar{\Xi}\left(z\right)^{-1}e_{i}\\
 & =\left[\bar{\Xi}\left(z\right)^{-1}Q\left(x\right)\bar{\Xi}\left(z\right)^{-1}\right]_{ji},\\
 & =\left[\text{vec}^{-1}\left(G_{z}\left(\bar{\xi}\left(s\right)\right)\right)\cdot Q\left(x\right)\cdot\text{vec}^{-1}\left(G_{z}\left(\bar{\xi}\left(s\right)\right)\right)\right]_{ji}
\end{align*}
and the result follows.
\end{proof}
Note Lemma \ref{lem:dz trace as a function of gradient} gives us
the relation in $\left(\ref{eq: second gradient as a function of dz}\right)$
for the sensing trajectory problem, and when in matrix form as in
the proof, gives a relationship that is simple to compute.

\section{Results}

\begin{figure}
\begin{centering}
\includegraphics[width=8cm]{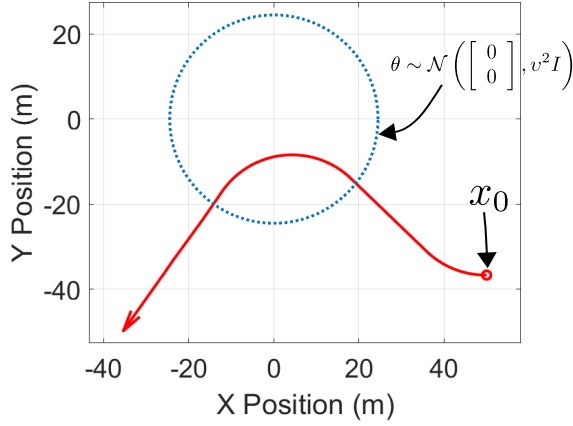}
\par\end{centering}
\caption{An optimal path computed for the first example, shown in red. In this
example, the aircraft is only using Doppler shift measurements. The
blue dashed circle is the 95\% error ellipse of the prior distribution
on $\theta$, which in this example represents the position of the
vehicle target.\label{fig:Single Opt. Path}}
\end{figure}
We consider a passive RF sensor that measures the Doppler frequency
shift in the carrier frequency, denoted as ${\mathcal{F}}$, arising
from the relative motion between transmitting vehicle and the receiver.
Note that we do not need to decode the underlying transmission, as
we are only tracking the carrier frequency. More details about the
derivation of this, as well as other sensor models can be found in
\cite{kirchner2020heterogeneous}.

We assume in this paper the sensor produces conditionally independent
measurements, each with a Gaussian distribution with mean $\mu_{\mathcal{F}}\left(\theta\right)$.
While the mean vector depends on the parameter of interest, $\theta$,
the covariance does not depend\footnote{It is not required that the covariance to be independent of $\theta$,
but it simplifies the example here.} on $\theta$and is given as $\Sigma_{\mathcal{F}}$. This gives a
closed form expression for $\left(\ref{eq:Inner Q(x;theta)}\right)$
for measurement ${\mathcal{F}}$, as
\begin{equation}
Q\left(x;\theta\right)=\left(\frac{\partial\mu_{\mathcal{F}}\left(\theta\right)}{\partial\theta}\right)^{\top}\Sigma_{\mathcal{F}}^{-1}\left(\frac{\partial\mu_{\mathcal{F}}\left(\theta\right)}{\partial\theta}\right),\label{eq:Gaussian FIM}
\end{equation}
where $\frac{\partial\mu_{\mathcal{F}}\left(\theta\right)}{\partial\theta}$
denotes the Jacobian matrix of $\mu_{\mathcal{F}}\left(\theta\right)$
\cite{malago2015information}. To estimate the expectation and find
the expression $\left(\ref{eq:Q(x)}\right)$, we choose a second-order
Taylor series expansion. Let $Q_{ij}\left(x;\theta\right)$ denote
the $i,j$-th element of the $\left(\ref{eq:Gaussian FIM}\right)$,
and $\theta$ is a random variable with mean $\mu_{\theta}$ and covariance
$\Sigma_{\theta}$. Then we approximate the element with a second
order Taylor expansion as
\begin{align*}
Q_{ij}\left(x;\theta\right)\approx & Q_{ij}\left(x;\mu_{\theta}\right)+\nabla Q_{ij}\left(x;\mu_{\theta}\right)^{\top}\left(\theta-\mu_{\theta}\right)\\
 & +\frac{1}{2}\left(\theta-\mu_{\theta}\right)^{\top}{\bf H}_{ij}\left(x;\mu_{\theta}\right)\left(\theta-\mu_{\theta}\right),
\end{align*}
where ${\bf H}_{ij}\left(x;\theta\right)$ is the hessian matrix of
$Q_{ij}\left(x;\theta\right)$ with respect to $\theta$. The expected
value is then found as
\begin{equation}
{\mathbb{E}}_{\theta}\left[Q_{ij}\left(x;\theta\right)\right]\approx Q_{ij}\left(x;\mu_{\theta}\right)+\frac{1}{2}\text{tr}\left(\Sigma_{\theta}{\bf H}_{ij}\left(x;\mu_{\theta}\right)\right).\label{eq:Taylor expectation approx}
\end{equation}
The closed-form gradient $\frac{\partial\mu_{\mathcal{F}}\left(\theta\right)}{\partial\theta}$
in $\left(\ref{eq:Gaussian FIM}\right)$ are found from \cite{kirchner2020heterogeneous},
while the Hessian values were found using the \noun{CasADI }toolbox
\cite{Andersson2019}.

In the example the parameters to be estimated, $\theta$, consist
of the $\left(X,Y\right)\in{\mathbb{R}}^{2}$ position of the target
vehicle. The prior distribution of $\theta$ is given as
\[
\theta\sim{\mathcal{N}}\left(\left[\begin{array}{c}
0\\
0
\end{array}\right],\upsilon^{2}I\right),
\]
where $\upsilon=10m$ is the standard deviation. The sensor measures
the Doppler shifts with noise standard deviation of $\Sigma_{\mathcal{F}}=1$.
The sensing aircraft is flying $1000\,m$ above the ground level where
the target vehicle is located and the turn rate is limited with $\omega_{\text{max}}=0.05$
rad/s.

Figure \ref{fig:Single Opt. Path} shows the optimal path from the
initial condition of $X\left(0\right)=50\,m$, $Y\left(0\right)=-36.6\,m$,
and $\psi\left(0\right)=-\pi$. The initial angle of $-\pi$ implies
the tracking aircraft is moving from right to left initially at $t=0$.
It can be seen in the figure that the optimal path begins with turning
maneuvers before traveling straight along a ray extending outward
from the center of the prior distribution of $\theta$. Conceptually,
travel along this ray will give maximum variation in Doppler shift,
but the early maneuvers are still necessary since multiple directions
of measurements are required to fully localize using only Doppler
measurements

 Figure \ref{fig:Doppler example} shows a series of
optimal paths generated with same initial conditions for $X\left(0\right)$
and $\psi\left(0\right)$, but with a variation in the initial condition,
$Y\left(0\right)$. The vertical initial condition, $Y\left(0\right)$,
were chosen uniformly from a range $\left[-50,50\right]$. While the
trajectories are different quantitatively from that of Figure \ref{fig:Single Opt. Path},
they share the same qualitative properties of an initial maneuver
to gain measurements in various directions before traveling away from
the prior belief, on a ray extending directly from the center. 
\begin{figure}
\begin{centering}
\includegraphics[width=8cm]{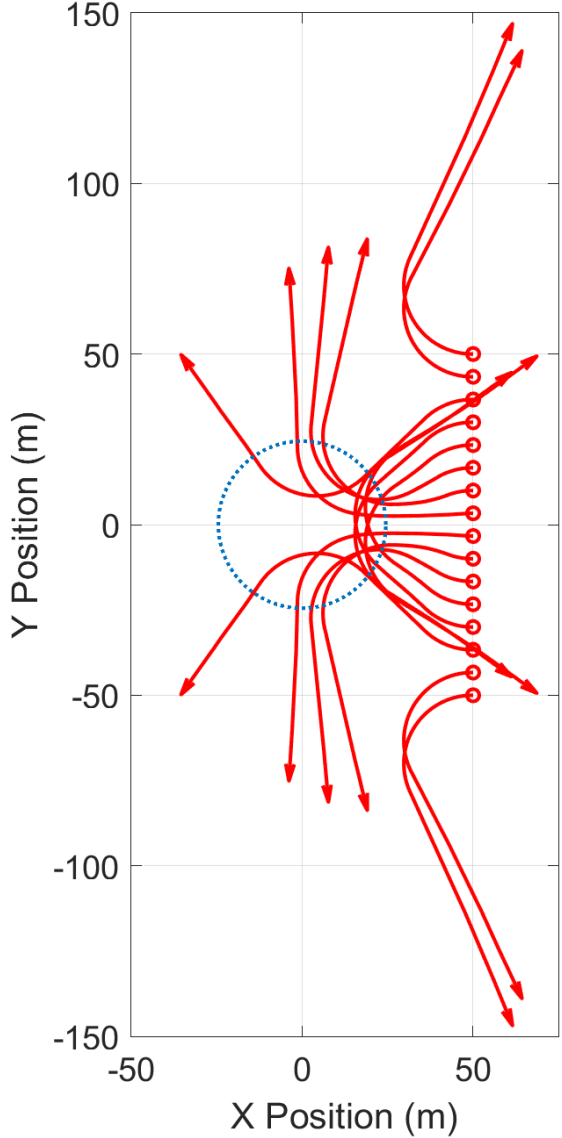}
\par\end{centering}
\caption{Here a series of optimal trajectories are shown in red from different
starting locations, with each vehicle starting out moving from right
to left. Same as in Fig. \ref{fig:Single Opt. Path}, the aircraft
is only using Doppler shift measurements. The blue circle is the 95\%
error ellipse of the prior distribution on $\theta$.\label{fig:Doppler example}}
\end{figure} 

\section{Conclusion}

We present a hybrid method of lines approach for solving a class of
Hamilton\textendash Jacobi PDEs that arise in the optimal placement
of sensors. This method provides for robustness, where needed, in
the $x$ subspace by using a classic grid approach with finite differencing.
It avoids a grid in the $z$ subspace and hence scales well with the
number of $z$ dimensions. We applied this to a trajectory optimization
problem where the goal is to find the trajectory that minimizes the
estimation error from the measurements collected along the calculated
path. Future work includes investigating metrics other than \noun{logdet}
such as the trace of the inverse and studying if the hybrid method
of lines approach can be generalized to a broader class of systems.

\appendices{}              

\section{\label{subsec:Regularity-Assumptions}Hamiltonian Regularity Assumptions}        
Let $n$ be the dimension of the augmented state variable $\chi$,
and denote by $\sigma:=\left(p,\lambda\right)^{\top}$, and with a
slight abuse of notation note that ${\mathcal{H}}\left(s,\chi,\sigma\right)={\mathcal{H}}\left(s,x,z,\sigma\right)={\mathcal{H}}\left(s,x,z,p,\lambda\right)$
and vice versa. We introduce a set of mild regularity assumptions:  
\begin{description}
\item [{(H1)}] The Hamiltonian
\[
\left[0,t\right]\times{\mathcal{X}}\times{\mathcal{Z}}\times{\mathbb{R}}^{n}\ni\left(s,x,z,p,\lambda\right)\mapsto{\mathcal{H}}\left(s,x,z,p,\lambda\right)\in{\mathbb{R}}
\]
is continuous.
\item [{(H2)}] There exists a constant $c>0$ such that for all $\left(s,x,z\right)\in\left[0,t\right]\times{\mathcal{X}}\times{\mathcal{Z}}$
and for all $\sigma',\sigma''\in{\mathbb{R}}^{n}$, the following inequalities
hold
\[
\left|{\mathcal{H}}\left(s,x,z,\sigma'\right)-{\mathcal{H}}\left(s,x,z,\sigma''\right)\right|\leq\kappa_{1}\left(\chi\right)\left\Vert \sigma'-\sigma''\right\Vert ,
\]
and
\[
\left|{\mathcal{H}}\left(s,x,z,{\bf 0}\right)\right|\leq\kappa_{1}\left(\chi\right),
\]
with $\kappa_{1}\left(\chi\right)=c\left(1+\left\Vert \chi\right\Vert \right)$.
\item [{(H3)}] For any compact set $M\subset{\mathbb{R}}^{n}$ there exists
a constant $C\left(M\right)>0$ such that for all $\chi',\chi''\in M$
and for all $\left(s,\sigma\right)\in\left[0,t\right]\times{\mathbb{R}}^{n}$
the inequality holds
\[
\left|{\mathcal{H}}\left(s,\chi',\sigma\right)-{\mathcal{H}}\left(s,\chi'',\sigma\right)\right|\leq\kappa_{2}\left(\sigma\right)\left\Vert \chi'-\chi''\right\Vert ;
\]
with $\kappa_{2}\left(\sigma\right)=C\left(M\right)\left(1+\left\Vert \sigma\right\Vert \right)$.
\item [{(H4)}] The terminal cost function
\[
{\mathbb{R}}^{n}\ni\chi\mapsto G\left(\chi\right)\in{\mathbb{R}},
\]
is continuous.
\end{description}
Next we present an important theorem on the existence and uniqueness
of viscosity solutions of the Hamilton\textendash Jacobi equation.
\begin{thm}[{\cite[Theorem II.8.1, p. 70]{subbotin1995generalized}}]
Let assumptions $\left(H1\right)-\left(H4\right)$ hold. Then there
exists a unique viscosity solution to $\left(\ref{eq:HJB Equation-1}\right)$.
\end{thm}

\section{Supporting Propositions}    
\begin{prop}
\label{prop:gradient in z with intial cond. is I}Let $\chi\in\mathcal{S}$,
then
\[
\frac{\partial}{\partial z}\xi\left(t;\chi,\bar{u}\left(\cdot\right)\right)=I.
\]
\end{prop}
\begin{proof}
By assumption, the terminal point of the state trajectory $\zeta\left(t;\chi,\bar{u}\left(\cdot\right)\right)$
is differentiable with respect to initial condition $\chi\in\mathcal{S}$.
Defining the Jacobin, for $s\in\left[0,t\right]$,
\begin{align*}
 & m\left(s\right):=\left[\begin{array}{cc}
m_{xx}\left(s\right) & m_{x,z}\left(s\right)\\
m_{zx}\left(s\right) & m_{zz}\left(s\right)
\end{array}\right]\\
 & =\left[\begin{array}{cc}
\bar{\gamma}_{x}\left(s\right) & \bar{\gamma}_{z}\left(s\right)\\
\bar{\xi}_{x}\left(s\right) & \bar{\xi}_{z}\left(s\right)
\end{array}\right]=\frac{\partial}{\partial\chi}\zeta\left(s;\chi,\bar{u}\left(\cdot\right)\right).
\end{align*}
We have from \cite[Chapter 5, Equation 3.23]{yong1999stochastic}
that $m\left(t\right)$ satisfies the following matrix equation almost
everywhere:
\[
\begin{cases}
\dot{m}\left(s\right)=\hat{f}_{\chi}\left(\bar{\zeta}\left(s;\chi,\bar{u}\left(\cdot\right)\right),\bar{u}\left(s\right)\right)m\left(s\right), & s\in\left[0,t\right],\\
m\left(0\right)=I.
\end{cases}
\]
From which the $m_{zz}$ partition is written as
\[
\begin{cases}
\dot{m}_{zz}\left(s\right)=\ell_{z}\left(\bar{\gamma}\left(s;\chi,\bar{u}\left(\cdot\right)\right)\right)m_{zz}\left(s\right), & s\in\left[0,t\right],\\
m_{zz}\left(0\right)=I.
\end{cases}
\]
Since $\ell$ does not depend on $z$, we have
\[
\dot{m}_{zz}\left(s\right)=0,\,\forall s\in\left[0,t\right],
\]
and the result follows.
\end{proof}

\acknowledgments
The authors would like to thank Levon Nurbekyan, with the Department
of Mathematics at UCLA, for providing a reference that assisted in
the proof of Theorem \ref{thm: ODE for grad of z}. This research was funded by the Office of Naval Research under Grant N00014-20-1-2093.

\bibliographystyle{IEEEtran}
\bibliography{Kirchner2023}

\thebiography
\begin{biographywithpic}
{Matthew R. Kirchner}{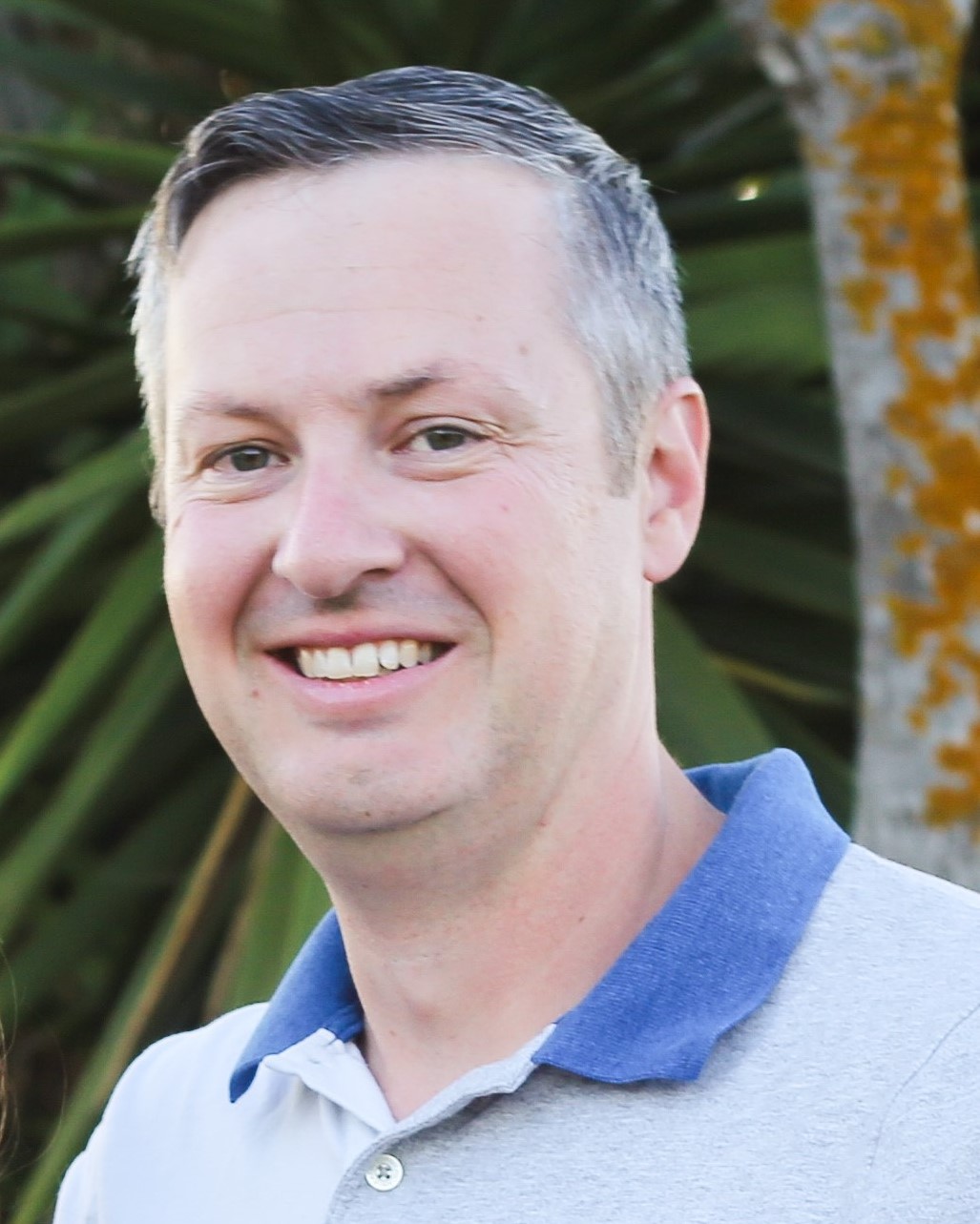}
received his B.S. in Mechanical Engineering from Washington State University in 2007 and his M.S. in Electrical Engineering from the University of Colorado at Boulder in 2013. In 2007 he joined the Naval Air Warfare Center Weapons Division in the Navigation and Weapons Concepts Develop Branch and in 2012 transferred into the Physics and Computational Sciences Division in the Research and Intelligence Department, Code D5J1000. He is currently a Ph.D. candidate studying Electrical Engineering at the University of California, Santa Barbara. His research interests include level set methods for optimal control, differential games, and reachability; multi-vehicle robotics; nonparametric signal and image processing; and navigation and flight control. He was the recipient of a Naval Air Warfare Center Weapons Division Graduate Academic Fellowship from 2010 to 2012; in 2011 was named a Paul Harris Fellow by Rotary International and in 2021 was awarded a Robertson Fellowship from the University of California in recognition of an outstanding academic record. Matthew is a student member of the IEEE.
\end{biographywithpic} 

\begin{biographywithpic}
{David Grimsman}{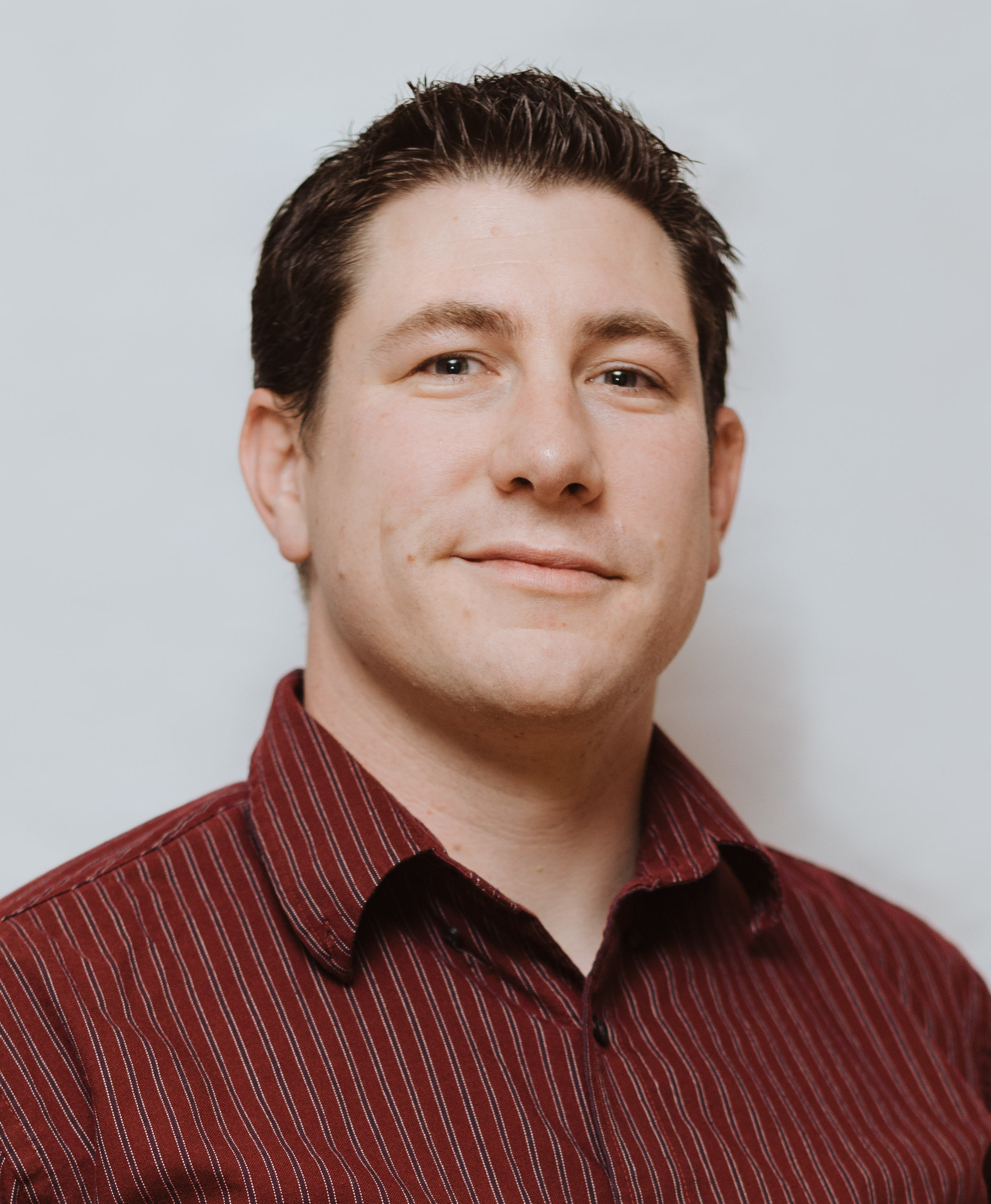}
is an Assistant Professor in
the Computer Science Department at Brigham
Young University. He completed BS in Electrical
and Computer Engineering at Brigham Young
University in 2006 as a Heritage Scholar, and
with a focus on signals and systems. After working for BrainStorm, Inc. for several years as a
trainer and IT manager, he returned to Brigham
Young University and earned an MS in Computer
Science in 2016. He then received his PhD in
Electrical and Computer Engineering from UC
Santa Barbara in 2021. His research interests include mulit-agent
systems, game theory, distributed optimization, network science, linear
systems theory, and security of cyberphysical systems.
\end{biographywithpic}

\begin{biographywithpic}
{Jo{\~a}o P. Hespanha}{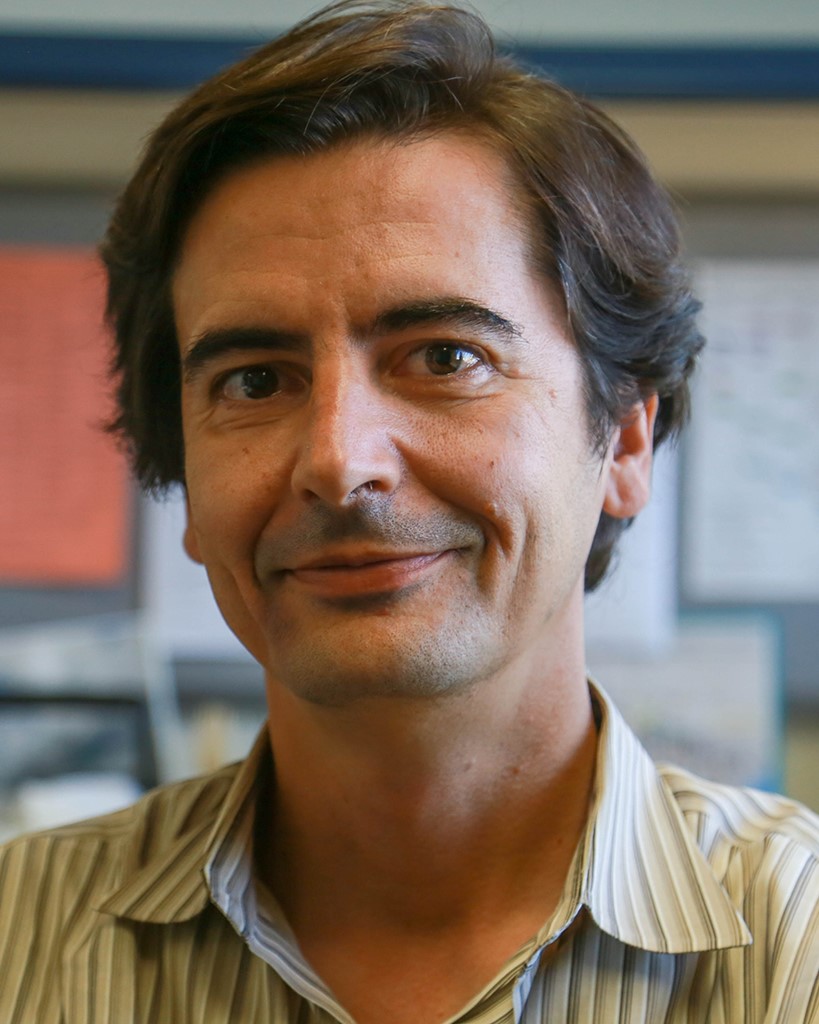}
received his Ph.D. degree in electrical engineering and applied science from Yale University, New Haven, Connecticut in 1998. From 1999 to 2001, he was Assistant Professor at the University of Southern California, Los Angeles. He moved to the University of California, Santa Barbara in 2002, where he currently holds a Professor position with the Department of Electrical and Computer Engineering. Dr. Hespanha is the recipient of the Yale University's Henry Prentiss Becton Graduate Prize for exceptional achievement in research in Engineering and Applied Science, the 2005 Automatica Theory/Methodology best paper prize, the 2006 George S. Axelby Outstanding Paper Award, and the 2009 Ruberti Young Researcher Prize. Dr. Hespanha is a Fellow of the IEEE and he was an IEEE distinguished lecturer from 2007 to 2013. His current research interests include hybrid and switched systems; multi-agent control systems; distributed control over communication networks (also known as networked control systems); the use of vision in feedback control; stochastic modeling in biology; and network security.

\end{biographywithpic}

\begin{biographywithpic}
{Jason R. Marden}{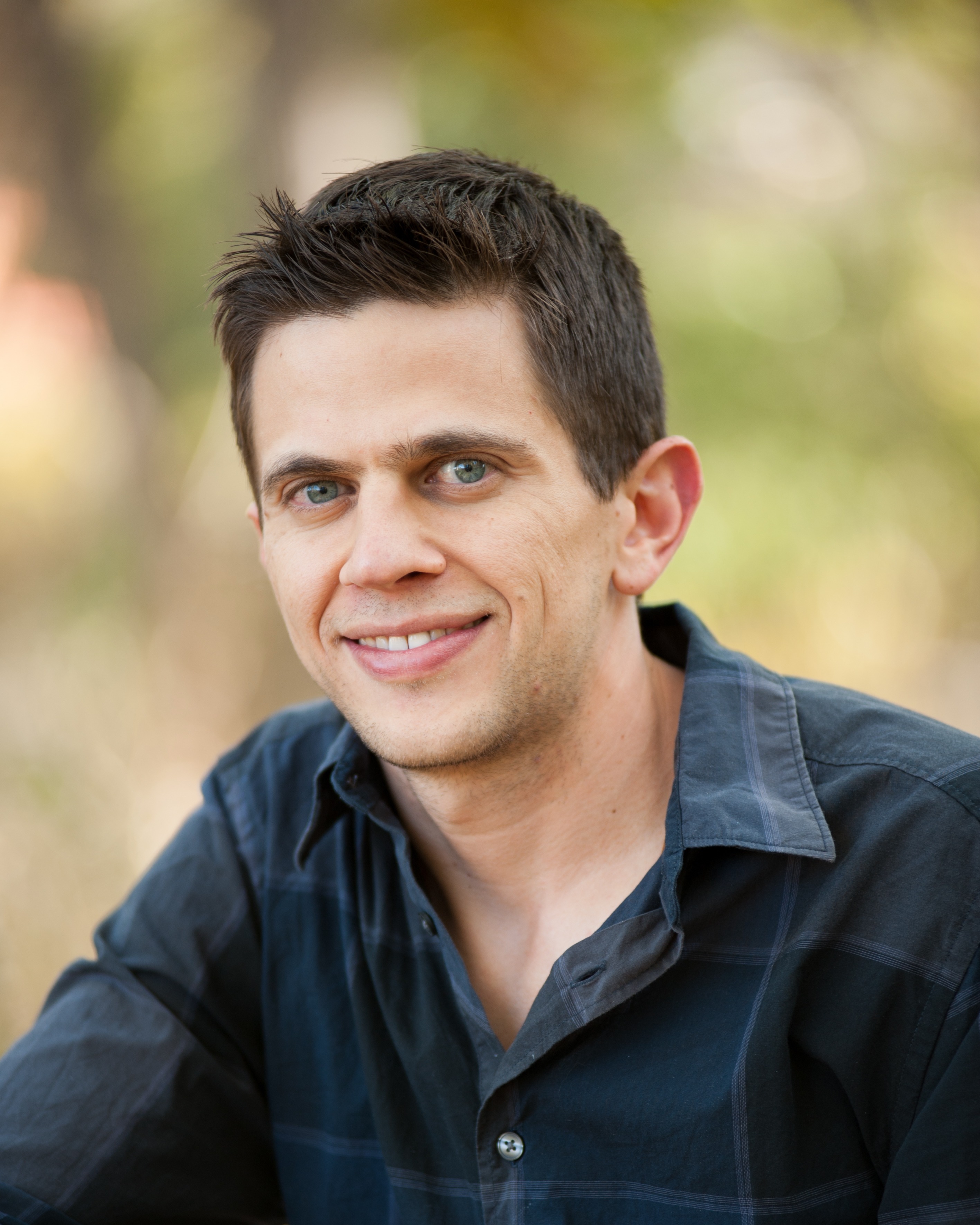}
is a professor in the Department of Electrical and Computer Engineering at the
University of California, Santa Barbara. He received the B.S. degree in 2001 and the Ph.D. degree in 2007 (under the supervision of Jeff S. Shamma), both in mechanical engineering from the University of California, Los Angeles, where he was awarded the
Outstanding Graduating Ph.D. Student in Mechanical Engineering. After graduating, he was
a junior fellow in the Social and Information Sciences Laboratory at the California Institute of Technology until 2010
and then an assistant professor at the University of Colorado until 2015.
He is a recipient of an ONR Young Investigator Award (2015), an NSF
Career Award (2014), the AFOSR Young Investigator Award (2012), the
SIAM CST Best Sicon Paper Award (2015), and the American Automatic
Control Council Donald P. Eckman Award (2012). His research interests
focus on game-theoretic methods for the control of distributed multiagent
systems.

\end{biographywithpic}
\balance

\end{document}